\documentclass[a4paper,superscriptaddress,nofootinbib,accepted=2019-10-31]{quantumarticle}

\pdfoutput=1

\usepackage{xkeyval}
\usepackage{etoolbox}
\usepackage{geometry}
\usepackage{xcolor}
\usepackage{fancyhdr} 
\usepackage{tikz}
\usepackage[colorlinks,linkcolor=black,citecolor=black,linktocpage,hypertexnames=true,pdfpagelabels]{hyperref}
\usepackage{ltxgrid}
\usepackage{ltxcmds}

\usepackage{amsmath}
\usepackage{amssymb}
\usepackage{amsthm}
\usepackage{bbm}
\usepackage{mathtools}
\usepackage{color}

\usepackage[capitalise]{cleveref}
\usepackage[all]{xy}
\usepackage{subcaption}
\usepackage{array}
\usepackage{enumerate}
\usepackage{tikz}
\usepackage{relsize}
\usepackage{diagbox}
\usepackage{tikz-cd} 

\newtheorem{theorem}{Theorem}
\newtheorem{lemma}[theorem]{Lemma}
\newtheorem{corollary}[theorem]{Corollary}
\newtheorem{proposition}[theorem]{Proposition}

\newtheorem{definition}[theorem]{Definition}

\newtheorem*{theorem-nn}{Theorem}

\theoremstyle{definition}

\newtheorem{example}[theorem]{Example}
\newtheorem{remark}[theorem]{Remark}

\renewcommand{\epsilon}{\varepsilon}
\newcommand{\N}{\mathbb{N}}

\newcommand{\R}{\mathbb{R}}
\newcommand{\M}{\mathcal{M}}

\newcommand{\bi}{\begin{itemize}}
\newcommand{\ei}{\end{itemize}}

\newcommand{\rank}{\textrm{rank}}
\newcommand{\psdrank}{\textrm{psd-rank}}
\newcommand{\purirank}{\textrm{puri-rank}}

\newcommand{\osr}{\textrm{osr}}
\newcommand{\hosr}{\textrm{hosr}}

\newcommand{\seprank}{\textrm{sep-rank}}

\newcommand{\tr}{\textrm{tr}}
\newcommand{\tsr}{\textrm{tsr}}

\usepackage{dsfont}

\newcommand{\be}{\begin{eqnarray}}
\newcommand{\ee}{\end{eqnarray}}
\newcommand{\ben}{\begin{enumerate}}
\newcommand{\een}{\end{enumerate}}
\newcommand{\im}{\item}
\newcommand{\ba}{\begin{array}}
\newcommand{\ea}{\end{array}}

\newcommand{\ra}{\rangle}
\newcommand{\la}{\langle}
\newcommand{\mc}{\mathcal}

\newcommand{\Her}{\mathrm{Her}}
\newcommand{\PSD}{\mathrm{PSD}}

\newcommand{\nn}{\nonumber}

\renewcommand{\theenumi}{(\roman{enumi})}

\def\thesection{\arabic{section}}
\def\thesubsection{\arabic{section}.\arabic{subsection}}
\makeatletter
\def\p@subsection{}
\def\p@subsubsection{}
\makeatother

\newcommand\xqed[1]{%
  \leavevmode\unskip\penalty9999 \hbox{}\nobreak\hfill
  \quad\hbox{#1}}
\newcommand\demo{\xqed{$\diamond$}}

\begin{document}

\title{Separability for mixed states with operator Schmidt rank two}

\author{Gemma De las Cuevas}
\address{Institute for Theoretical Physics, Technikerstr.\ 21a,  A-6020 Innsbruck, Austria}
\author{Tom Drescher}
\address{Department of Mathematics, Technikerstr.\ 13,  A-6020 Innsbruck, Austria}
\author{Tim Netzer}
\address{Department of Mathematics, Technikerstr.\ 13,  A-6020 Innsbruck, Austria}

\begin{abstract}
The operator Schmidt rank is the minimum number of terms required to express a state as a sum of elementary tensor factors. 
Here we provide a new proof of the fact that any bipartite mixed state with operator Schmidt rank two is separable, and can be written as a sum of two positive semidefinite matrices per site.  
Our proof uses results from the theory of free spectrahedra and operator systems, 
and illustrates  the use of a connection between decompositions of quantum states and decompositions of nonnegative matrices.
In the multipartite case, 
we prove that any Hermitian Matrix Product Density Operator (MPDO) of bond dimension two is separable, 
and  can be written as a sum of at most four positive semidefinite matrices per site. 
This implies that these states can only contain classical correlations, and very few of them, as measured by the entanglement of purification.
In contrast, MPDOs of bond dimension three can contain an unbounded amount of classical correlations. 
\end{abstract}

\maketitle

\section{Introduction}

Entanglement is an essential ingredient in many applications in quantum information processing and quantum computation \cite{Ni00,Ho09b}. 
Mixed states which are not entangled are called \emph{separable}, 
and they may contain only classical correlations---in  contrast to entangled states, which contain quantum correlations. 
As many other problems in theoretical physics and elsewhere, the problem of determining whether a state is entangled or not 
is  NP-hard \cite{Gu03,Gh10}. 
This does not  prevent the existence of multiple separability criteria \cite{Ho09b}. 
One example are criteria based on the rank of a bipartite mixed state $0\leqslant \rho \in \mc{M}_{d_1}\otimes \mc{M}_{d_2}$ (where $\mc{M}_d$ denotes the set of complex matrices of size $d\times d$, and $A\geqslant 0$ denotes that $A$ is positive semidefinite) \cite{Kr00,Ho00}. 
Another example is the realignment criterion \cite{Ch03,Ru03c}, which is based on the operator Schmidt decomposition of the state.

Here, we focus on the \emph{operator Schmidt rank} of $\rho$, which is the minimum $p$ so that 
\be\label{eq:rhobipar}
\rho =\sum_{\alpha=1}^p A_\alpha \otimes B_\alpha , 
\ee
where $A_\alpha \in \mc{M}_{d_1}$ and $B_\alpha\in\mc{M}_{d_2}$, i.e.\ these matrices  need not fulfill any conditions of Hermiticitiy or positivity. Clearly, a product state (i.e.\ a state without classical or quantum correlations) has $p=1$, 
but it is generally not clear what kind of correlations a  state with $p>1$ has.

The first main result of this paper is that, if a state has operator Schmidt rank  two, then it is separable.  Moreover, it can be written as a sum of only two positive semidefinite matrices at each site. 
Our proof is constructive and gives a method to obtain the positive semidefinite matrices, as we  will explicitly show. 
To state our result formally, denote by $\textrm{PSD}_d$ the set of 
 $d\times d$ positive semidefinite matrices.

\begin{theorem}[Bipartite case]\label{thm:bipartite}
Let $\rho$ be a  positive semidefinite matrix
 $0\leqslant \rho\in \mc{M}_{d_1}\otimes \mc{M}_{d_2} $ 
where $d_1,d_2$ are arbitrary, such that 
$$
\rho = 
\sum_{\alpha=1}^2 
A_{\alpha} \otimes B_\alpha 
$$
where $A_\alpha\in \mc{M}_{d_1}$ and $B_\alpha\in \mc{M}_{d_2}$. 
Then $\rho$ is separable and can be written as 
$$
\rho = 
\sum_{\alpha=1}^2 
\sigma_{\alpha}\otimes \tau_\alpha
$$
where  $\sigma_\alpha \in \PSD_{d_1}$ and 
$\tau_\alpha \in \PSD_{d_2}$. 
\end{theorem}

To the best of our knowledge, this result was first proven by Cariello in \cite[Lemma 4.5]{Ca14b} (see also \cite[Theorem 3.44]{Ca17b} and this blog post \cite{Na14}), and extended to the case that one side is infinite dimensional in \cite[Theorem 3.3]{Gi18b}. (See also the related inequalities \cite[Theorem 5.2, Theorem 5.3]{Ca14b}.) 
Here we provide a new proof for this result, using other techniques and a new context. 
The techniques are based on  free spectrahedra, which form a relatively young field of study within convex algebraic geometry, and whose connections with decompositions of quantum states are explored in \cite{Ne19}. 
In particular, the core result of this paper is a relation between  separable states and minimal operator systems  (\cref{thm:main}), and its consequences  (\cref{cor:main}), which we leverage to prove \cref{thm:bipartite}. 
(See Ref.\ \cite{Bl18} for other recent connections between free spectrahedra and quantum information.) 
The context of this work is  a connection between decompositions of quantum states and decompositions of nonnegative matrices presented in \cite{De19}. 
\cref{thm:bipartite} is inspired by an analogous result for nonnegative matrices, and can be seen as a generalisation thereof, as we will show in \cref{ssec:context}.

What happens in the multipartite case? For pure states, multipartite entanglement  features many novel properties with respect to the bipartite case, see e.g.\ \cite{Sa18}.
For a mixed state that  describes the state of a spin chain in one spatial dimension, 
$0\leqslant \rho\in \mc{M}_{d_1}\otimes \mc{M}_{d_2} \otimes \cdots \otimes\mc{M}_{d_n}$,
 the multipartite analogue of Eq.\ \eqref{eq:rhobipar} is the Matrix Product Density Operator (MPDO) form \cite{Ve04d,Zw04}, 
\be
\rho = \sum_{\alpha_1,\ldots, \alpha_{n-1}=1}^D 
A^{[1]}_{\alpha_1} \otimes A^{[2]}_{\alpha_1,\alpha_2}\otimes 
 \cdots \otimes A^{[n]}_{\alpha_{n-1}} ,
\nn
\ee 
where $A^{[1]}_\alpha \in \mc{M}_{d_1}$, $A^{[2]}_{\alpha,\beta} \in \mc{M}_{d_2}$, \ldots, $A^{[n]}_{\alpha} \in \mc{M}_{d_n}$. 
The minimal such $D$ is also called the  operator Schmidt rank of $\rho$.  
If, additionally, we require that each of the local matrices is Hermitian, i.e.\ $A^{[1]}_\alpha \in \Her_{d_1}$, $A^{[2]}_{\alpha,\beta} \in \Her_{d_2}$, \ldots, $A^{[n]}_{\alpha} \in \Her_{d_n}$, where $\textrm{Her}_d$ denotes the set of  $d\times d$ Hermitian matrices, then the minimal such $D$ is called the \emph{Hermitian operator Schmidt rank} of $\rho$ \cite{De19}.

One of the problems of the  MPDO form is that the amount of quantum or classical correlations of the state (as measured by the entanglement of purification \cite{Te02}) is not upper bounded by the operator Schmidt rank \cite{De13c,De19}. 
Instead, the amount of correlations is upper bounded by the \emph{purification rank} of $\rho$, 
which is the smallest operator Schmidt rank of a matrix $L$ that satisfies $\rho = LL^\dagger$, i.e.\ that makes explicit the positivity of $\rho$ \cite{De13c,De19}.
Since the purification rank can be arbitrarily larger than the operator Schmidt rank \cite{De13c}, 
the latter generally does not tell us anything about the correlations of the state. 
In particular, there are families of states with operator Schmidt rank three whose purification rank diverges \cite{De13c}.
A similar thing happens for the Hermitian operator Schmidt rank: 
while it is lower bounded by the operator Schmidt rank, it does not allow us to upper bound the amount of correlations of the state. 
The exception to this rule is again the trivial case of product states, in which case  the operator Schmidt rank, its Hermitian counterpart and the purification rank are all one.

The second main result of this paper is that if a multipartite state has Hermitian operator Schmidt rank two, then it is separable. 
Moreover, it admits a separable decomposition of bond dimension two, which means that it can be written as a sum of at most four positive semidefinite matrices per site. Formally:

\begin{theorem}[Multipartite case]\label{thm:mainmpo}
Let $\rho$ be  a positive semidefinite matrix, 
$
0\leqslant \rho\in \mc{M}_{d_1}\otimes \mc{M}_{d_2} \otimes \cdots \otimes\mc{M}_{d_n}
$ 
where $d_1,d_2,\ldots,d_n$ are arbitrary, such that
\be
\rho = \sum_{\alpha_1,\ldots, \alpha_{n-1}=1}^2 
A^{[1]}_{\alpha_1} \otimes A^{[2]}_{\alpha_1,\alpha_2}\otimes 
 \cdots 
 %\otimes A^{[n-1]}_{\alpha_{n-2},\alpha_{n-1}} 
 \otimes A^{[n]}_{\alpha_{n-1}}
\nn
\ee 
where $A^{[1]}_\alpha \in \Her_{d_1}$, $A^{[2]}_{\alpha,\beta} \in \Her_{d_2}$, \ldots, $A^{[n]}_{\alpha} \in \Her_{d_n}$.
Then $\rho$ is separable and can be written as 
$$
\rho = 
\sum_{\alpha_1,\ldots, \alpha_{n-1}=1}^2 
\sigma^{[1]}_{\alpha_1} \otimes \sigma^{[2]}_{\alpha_1,\alpha_2}\otimes 
%\sigma^{[3]}_{\alpha_2,\alpha_3} \otimes 
\cdots \otimes \sigma^{[n]}_{\alpha_{n-1}}
$$
where  $\sigma^{[1]}_\alpha \in \PSD_{d_1}$, 
$\sigma^{[2]}_{\alpha,\beta} \in \PSD_{d_2}$, 
$\ldots$, $\sigma^{[n]}_{\alpha} \in \PSD_{d_n}$.
\end{theorem}

This situation is summarised in \cref{tab:bipartite}.

\begin{table}[ht]\centering
\begin{tabular}{c | c | c | c}
States & Product  & Separable & \\
\hline
$\hosr$ & 1 & 2 & 3\\
$\purirank$ &1 & 2 & unbounded\\ 
$\seprank$ &1& 2 & unbounded
\end{tabular}
\caption{Multipartite mixed states $\rho$ with hermitian operator Schmidt rank  (cf.\ \cref{def:mpdo})  $\hosr(\rho)=1$ are product states, and the 
separable rank $\seprank$ (\cref{def:sep}) and 
purification rank $\purirank$ (\cref{def:lp}) need also be 1. 
If $\hosr(\rho)=2$ then $\rho$ must be separable and  the other two ranks also need be 2 (\cref{thm:mainmpo}).
If $\hosr(\rho)=3$ and $\rho$ is separable, then the other two ranks may diverge with the system size \cite{De13c}. }
\label{tab:bipartite}
\end{table}

We remark that \cite[Corollary 4.8]{Ca14b} showed that if a state $\rho$ has  tensor rank 2, then it is separable. Recall that the tensor rank, $\tsr(\rho)$, is the minimal  $D$ required to express $\rho$ as
$$
\rho=\sum_{\alpha=1}^D A_{\alpha}^{[1]} \otimes A_{\alpha}^{[2]}\otimes \ldots A_{\alpha}^{[n]}. 
$$
\cref{thm:mainmpo}, in contrast, shows that if the Hermitian operator Schmidt rank of a state $\rho$ is 2, then $\rho$ is separable and its separable rank is 2  
(the latter will be defined in \cref{def:sep}).
Since both the hermitian operator Schmidt rank ($\hosr$) and the tensor rank ($\tsr$) are  lower bounded by the operator Schmidt rank ($\osr$),  
$$ 
\osr\leq \hosr,\quad \osr\leq \tsr,  
$$
our \cref{thm:mainmpo} does not follow from \cite[Corollary 4.8]{Ca14b}
(For more relations between the ranks of several types of tensor decompositions, see \cite{De19d}).
Note also that  \cref{thm:mainmpo} does not only show that $\rho$ is separable but also that the separable rank (see \cref{def:sep}) is 2.

%%------------------------
This paper is structured as follows.  
In \cref{sec:pre} we present the preliminary material needed for this work. 
In \cref{sec:sep} we prove the core result, \cref{thm:main}. 
In \cref{sec:bipartite} we focus on the bipartite case: 
we prove \cref{thm:bipartite}, provide a method to obtain the separable decomposition,  and explore several implications of this result. 
In \cref{sec:mpdo} we focus on the multipartite case: 
we prove \cref{thm:mainmpo}, and explore its implications. 
Finally, in  \cref{sec:concl} we conclude and present an outlook.

%%==========================
\section{Preliminaries}
\label{sec:pre}

Here we present the preliminary material needed to prove the results of this work.
First we will define the relevant decompositions of mixed states (\cref{ssec:decomp}),
and finally define free spectrahedra and minimal operator systems (\cref{ssec:fs}).

We use the following notation. 
The set of positive semidefinite matrices, complex Hermitian matrices, and complex matrices of size $r\times r$ is denoted $\textrm{PSD}_r$, $\textrm{Her}_r$ and $\M_r$, respectively. 
We denote that a matrix $A$ is positive semidefinite (psd) by $A\geqslant 0$. The identity matrix is denoted by $I$.
We will also often ignore normalizations, as they do not modify any of the ranks  \cite{De19}. 
For this reason, we often refer to positive semidefinite matrices as states. 

%%================================
\subsection{Decompositions of mixed states}
\label{ssec:decomp}

Throughout this paper we consider a multipartite positive semidefinite matrix $0\leqslant \rho\in \mc{M}_{d_1}\otimes \mc{M}_{d_2} \otimes \cdots \otimes\mc{M}_{d_n}$. We now review the MPDO, the Hermitian MPDO form, the separable decomposition and the local purification of $\rho$. They are all discussed and characterized in Ref.\ \cite{De19}.

\begin{definition}\label{def:mpdo}
The \emph{Matrix Product Density Operator (MPDO)} form \cite{Ve04d,Zw04} of $\rho$ is given by 
\be
\rho = \sum_{\alpha_1,\ldots,\alpha_{n-1}=1}^D 
A^{[1]}_{\alpha_1} \otimes 
A^{[2]}_{\alpha_1,\alpha_2} \otimes 
%A^{[3]}_{\alpha_2,\alpha_3}\otimes 
\cdots \otimes
A^{[n]}_{\alpha_{n-1}} ,
\label{eq:mpdo}
\ee
where $A^{[l]}_{\alpha}\in\M_{d_l}$ for $l=1,n$, and
$A^{[l]}_{\alpha,\beta}\in\M_{d_l}$ for $1<l<n$. 
The minimum such $D$ is called the operator Schmidt rank of $\rho$, denoted $\osr(\rho)$. 
\end{definition}

The MPDO form can also be defined for matrices which are not Hermitian or positive semidefinite. In this case, it is simply called the Matrix Product Operator form (since $\rho$ is no longer a density operator), and the corresponding minimum dimension is also called the operator Schmidt rank. 

We also need to consider a modified version of the MPDO form, in which the local matrices are Hermitian. 
\begin{definition}\label{def:hermpdo}
The \emph{Hermitian Matrix Product Density Operator (Hermitian MPDO)} form  \cite{De19} of $\rho$ is given by 
\be
\rho = \sum_{\alpha_1,\ldots,\alpha_{n-1}=1}^D 
A^{[1]}_{\alpha_1} \otimes 
A^{[2]}_{\alpha_1,\alpha_2} \otimes 
%A^{[3]}_{\alpha_2,\alpha_3}\otimes 
\cdots \otimes
A^{[n]}_{\alpha_{n-1}} ,
\label{eq:hmpdo}
\ee
where $A^{[l]}_{\alpha}\in\Her_{d_l}$ for $l=1,n$, and
$A^{[l]}_{\alpha,\alpha'}\in\Her_{d_l}$ for $1<l<n$. 
The minimum such $D$ is called the Hermitian operator Schmidt rank of $\rho$, denoted $\hosr(\rho)$. 
\end{definition}

In this paper, we say that  $\rho$ is \emph{separable } if it admits a separable decomposition. That this is a meaningful definition is shown in \cite{De19}.  

\begin{definition}\label{def:sep} 
A \emph{separable decomposition} of $\rho$  is a form where 
\be
\rho = \sum_{\alpha_1,\ldots,\alpha_{n-1}=1}^D \sigma^{[1]}_{\alpha_{1}} \otimes
\sigma^{[2]}_{\alpha_{1},\alpha_2}  \otimes \cdots \otimes 
\sigma^{[n]}_{\alpha_{n-1}}  
\label{eq:sep}
\ee
where each of these matrices is psd, i.e.\ $\sigma^{[1]}_{\alpha}\in \PSD_{d_1}$, $\sigma^{[n]}_{\alpha}\in\PSD_{d_n}$,
  and $\sigma^{[j]}_{\alpha,\beta}\in\PSD_{d_j}$ for $1<j<n$.
The minimal such $D$ is called the separable rank of $\rho$, denoted $\textrm{sep-rank}( \rho) $. 
\end{definition}

We now review the local purification form. 

\begin{definition}\label{def:lp} 
The \emph{local purification form} of  $\rho$ is an expression $\rho=LL^\dagger$ where $L$ is in Matrix Product Operator form. The minimal operator Schmidt rank of such an $L$ is called the purification rank of $\rho$, denoted $\purirank(\rho)$:
$$
\purirank(\rho) =\min\{\osr(L) | \rho = LL^\dagger\} .
$$
\end{definition}

Finally we  review some results about these ranks. 
\begin{proposition}[\cite{De19}] \label{thm:basic}
Let 
$$
0\leqslant \rho\in \mc{M}_{d_1}\otimes \mc{M}_{d_2} \otimes \cdots \otimes\mc{M}_{d_n} . 
$$ 
 Then 
\begin{itemize} 
\item[(i)] $\osr(\rho) = 1\iff \purirank(\rho) =1 \iff \seprank(\rho)=1.$
\item[(ii)] $\osr(\rho)\leq \purirank(\rho)^2$, and this bound is tight for pure states. 
\item[(iii)] If $\rho$ is separable, then 
$ \osr(\rho)\leq \seprank(\rho)$, and
$ \purirank(\rho)\leq \seprank(\rho)$.
\end{itemize}
\end{proposition}

It is also known that there is a separation between $\osr(\rho)$ and $\purirank(\rho)$, and between  $\purirank(\rho)$ and  $\seprank(\rho)$ \cite{De19}.

%%===================
\subsection{Free spectrahedra}
\label{ssec:fs}

Here we review the definitions of (free) spectrahedra, which are discussed more thoroughly in \cite{Ne19}, 
as well as the definition of minimal operator system. 

\begin{definition} 
A {\it spectrahedron} is a set of the following form 
$$
\textrm{S}(A_1,\ldots, A_m) = \{(b_1,\ldots, b_m)\in \R^m | \sum_{i=1}^m b_i A_i \geqslant 0\}.
$$
where $A_1,\ldots A_m \in \Her_s$ are  Hermitian matrices.
\end{definition}

The free spectrahedron is its non-commutative generalization:  

\begin{definition}
For $A_1,\ldots, A_m\in \Her_s$  and $r\in\N$ we define 
\be
\textrm{FS}_r(A_1,\ldots, A_m):= && \{ (B_1,\ldots, B_m)\in{\rm Her}_r^m\mid  \nn\\ 
&&\sum_{i=1}^m A_i\otimes B_i\geqslant 0\}
\ee
and call it the {\it $r$-level} of the free spectrahedron defined by $A_1,\ldots, A_m$. The collection 
$$
\textrm{FS}(A_1,\ldots, A_m):=\left( \textrm{FS}_r(A_1,\ldots, A_m) \right)_{r\geq 1}
$$ 
is called the {\it free spectrahedron} defined by $A_1,\ldots, A_m$.
\end{definition}

Note that the free spectrahedron at level 1 coincides with the spectrahedron. 

Finally we review the definition of minimal operator system  (see, for example, Ref.\ \cite{Fr16} for more details). 

\begin{definition}
Let $C\subseteq \R^m$ be a closed salient convex cone 
(i.e.\ a closed cone that does not contain a nontrivial subspace) with nonempty interior.
Define 
$$
C_r^{\min} := \left\{\sum_{i=1}^n c_i\otimes P_i \mid n\in\N,  c_i \in C, P_i \in \PSD_{r}\right\}
$$ 
Note that since the $c_i$ are from $\R^m$, elements from $C_r^{\min}$ can be understood as $m$-tuples of Hermitian matrices of size $r$ (exactly as elements from the $r$-level of a free spectrahedron).

The minimal operator system is defined as 
$$
C^{\min} = \left(C_r^{\min} \right)_{r\geq 1} .
$$
\end{definition}

%%====================================
\section{Separable states and minimal operator systems}
\label{sec:sep}

In this section we prove the core result of this paper, namely \cref{thm:main}, which is a relation between separable states and  minimal operator systems. We will extract several consequences thereof in  \cref{cor:main}, which  will be exploited later on to prove \cref{thm:bipartite} and \cref{thm:mainmpo}.

\begin{theorem} 
\label{thm:main}
Let $P_1,\ldots, P_m\in{\rm Her}_s$ be $\mathbb R$-linearly independent, and $Q_1,\ldots, Q_m\in{\rm Her}_t$ with 
\be
\rho=\sum_{i=1}^m P_i\otimes Q_i \geqslant 0. 
\label{eq:rhothmmain}
\ee
Define  
$$
C:={\rm FS}_1(P_1,\ldots, P_m)\subseteq \mathbb R^m
$$ 
and 
$$
 V:={\rm span}_{\mathbb R}\left\{P_1,\ldots, P_m \right\}\subseteq{\rm Her}_s.
$$
Then   the following are equivalent:
\begin{enumerate}
\item \label{thm:1}
 $(Q_1,\ldots, Q_m)\in C^{\min}_t$
\item  \label{thm:2}
$\rho$ admits a separable decomposition 
\be
\rho=\sum_{j=1}^n \tilde P_j\otimes \tilde Q_j
\label{eq:rhothmsep}
\ee
with $\tilde P_j\in V$ for all $j$.
\end{enumerate}
\end{theorem}

\begin{proof} 
\ref{thm:1}$\Rightarrow$\ref{thm:2}:\ By definition of $C^{\min}$ there are $v_1,\ldots, v_n \in C$ and $H_1,\ldots, H_m\in{\rm Her}_t$ with all $H_i\geqslant 0$ 
and 
%$$
%(Q_1,\ldots, Q_m)=\sum_{j=1}^n v_j \otimes H_j
%$$ 
%which just means that 
$$
Q_i=\sum_{j=1}^n v_{ji}H_j
$$ 
for all $i$. 
Substituting in the expression of $\rho$ (Eq. \eqref{eq:rhothmmain}), we obtain 
$$
\rho=\sum_{j=1}^n \left(\sum_{i=1}^m v_{ji}P_i\right)\otimes H_j.
$$ 
From $v_j\in C={\rm FS}_1(P_1,\ldots, P_m)$ it follows that $\sum_i v_{ji}P_i\geqslant 0$ for all $j$,  which proves \ref{thm:2}.

\ref{thm:2} $\Rightarrow$ \ref{thm:1}:\ 
Assume that $\rho$ is of form \eqref{eq:rhothmsep}  with all $\tilde P_j,\tilde Q_j\geqslant 0$ and 
$$
\tilde P_j=\sum_{i=1}^m v_{ji}P_i
$$ 
for some $v_j=(v_{j1},\ldots, v_{jm})\in\mathbb R^m$. 
Substituting in the expression of $\rho$ [Eq.\ \eqref{eq:rhothmsep}] we obtain that 
$$
\rho=
\sum_{i=1}^m P_i \otimes \left(\sum_{j=1}^n v_{ji}\tilde Q_j\right).
$$ 
Comparing this with expression \eqref{eq:rhothmmain}  and using that $P_1,\ldots, P_m$ are linearly independent, we obtain $$Q_i=\sum_{j=1}^n v_{ji}\tilde Q_j$$ for all $i$, i.e.\ 
$$
(Q_1,\ldots, Q_m)=\sum_{j=1}^n v_j\otimes\tilde Q_j\in C^{\min}_t.
$$
\end{proof}

\begin{remark}
In Theorem \ref{thm:main}, ``\ref{thm:2} $\Rightarrow$\ref{thm:1}"\ fails without the assumption $\tilde P_j\in V$ for all $j$. To see this, consider the following example taken from \cite{He13}. 
Consider 
$$
P_1=\left(\begin{array}{ccc}1 & 0 & 0 \\0 & 1 & 0 \\0 & 0 & 1\end{array}\right),\ 
P_2=\left(\begin{array}{ccc}0 & 1 & 0 \\1 & 0 & 0 \\0 & 0 & 0\end{array}\right),\  
P_3=\left(\begin{array}{ccc}0 & 0 & 1 \\0 & 0 & 0 \\1 & 0 & 0\end{array}\right)
$$  
  $$Q_1=\left(\begin{array}{cc}1 & 0 \\0 & 1\end{array}\right),\ Q_2=\left(\begin{array}{cc}\tfrac12 & 0 \\0 & 0\end{array}\right),\ Q_3=\left(\begin{array}{cc}0 & \tfrac34 \\\tfrac34 & 0\end{array}\right)$$ and 
  $$
  \rho:=\sum_{i=1}^3 P_i\otimes Q_i
  $$ 
   One easily checks that $\rho\geqslant 0$, and since all appearing matrices are symmetric, $\rho$ passes the separability test of partial transposition \cite{Ho09b}. This test is however sufficient for separability in the $3\times 2$ case, so $\rho$ is separable.

Now consider the Pauli matrices 
 $$
 \sigma_z = \left(\begin{array}{cc}1 & 0  \\0 & -1\end{array}\right) , \quad 
 \sigma_x = \left(\begin{array}{cc}0 & 1  \\1 & 0\end{array}\right),
 $$ and 
note that 
$$
C={\rm FS}_1(P_1,P_2,P_3)={\rm FS}_1(I_2,\sigma_z,\sigma_x)\subseteq\mathbb R^3
$$
 is the circular cone, i.e.\ the cone defined by the real numbers $a,b,c$ such that
 $b^2+c^2\leq a^2$, and  $a\geq 0$. 
However, on the other hand, 
 $$
 I_2\otimes Q_1+ \sigma_z\otimes Q_2 +\sigma_3\otimes Q_3\ngeqslant 0,
 $$ 
 which implies that $(Q_1,Q_2,Q_3)\notin C^{\min}_2.$
\end{remark}

\begin{corollary} 
\label{cor:main}
Let $P_1,\ldots, P_m\in{\rm Her}_s$ be $\mathbb R$-linearly independent, and set  $$C:={\rm FS}_1(P_1,\ldots, P_m)\subseteq\mathbb R^m.$$

\begin{enumerate}%[ref=(\roman{enumi})]
\item \label{cor:1}
If  ${\rm FS}(P_1,\ldots, P_m) =C^{\min}$ then for any choice of  $Q_1,\ldots, Q_m\in{\rm Her}_t$ with 
$$\rho=\sum_{i=1}^m P_i\otimes Q_i\geqslant 0,$$ the state $\rho$ is separable.

\item \label{cor:2}
 If $C$ is a simplex cone, then for any choice of  $Q_1,\ldots, Q_m\in{\rm Her}_t$ with 
 $$\rho=\sum_{i=1}^m P_i\otimes Q_i\geqslant 0,$$ 
 the state $\rho$ is separable. Furthermore, it admits a separable decomposition with at most $m$ terms, that is, $\seprank(\rho)\leq  m$. 

\item \label{cor:3} 
If $m=2$, then for any choice of  $Q_1, Q_2\in{\rm Her}_t$ with 
 $$\rho=P_1\otimes Q_1 + P_2\otimes Q_2\geqslant 0,$$ 
  the state $\rho$ is separable and admits a separable decomposition with at most $2$ terms, i.e.\ $\seprank(\rho)\leq 2$.
\end{enumerate}
\end{corollary}

\begin{proof}
\ref{cor:1}:
The condition $\rho\geqslant 0$ just means $$(Q_1,\ldots, Q_m)\in {\rm FS}_t(P_1,\ldots, P_m)=C_t^{\min}.$$ So the result follows from Theorem \ref{thm:main}.

\ref{cor:2}:
From \cite{Fr16} we know that  if $C$ is a simplex cone, then 
$$
C^{\min}={\rm FS}(P_1,\ldots, P_m) 
$$ 
 is automatically true. So the result follows from \ref{cor:1}. Note that the number of terms in the constructed separable decomposition is $m$, since in the proof of  \cref{thm:main} ``\ref{thm:1} $\Rightarrow$ \ref{thm:2}" we can choose the $v_i$ as the extreme rays of $C$, of which  there are $m$ many.

\ref{cor:3}: For  $C={\rm FS}_1(P_1,P_2)\subseteq \mathbb R^2$ there are three possibilities. First, $C$ might be  a simplex cone, so the  claim follows from \ref{cor:2}. Second, $C$ might be a single ray. In this case one easily checks that $\rho$ is in fact a product state, so $\seprank(\rho)= 1$. Finally $C=\{0\}$ implies $\rho=0$, so everything is trivial.
\end{proof}

%%===========================================
\section{Bipartite states with operator Schmidt rank two}
\label{sec:bipartite}

In this section we focus on bipartite states which have operator Schmidt rank two.
We will first  prove \cref{thm:bipartite} (\cref{ssec:proofthm1}), and then 
 provide a method to obtain the separable decomposition of such bipartite states (\cref{ssec:method}). 
 The method closely follows the proof and thereby illustrates our construction. 
Then we will put \cref{thm:bipartite} in the broader context of decompositions of nonnegative matrices, and derive some implications for the ranks (\cref{ssec:context}). 
Finally we will derive further implications for quantum channels and the entanglement of purification (\cref{ssec:otherimpl}).

\subsection{Proof of \cref{thm:bipartite}}
\label{ssec:proofthm1}

In this section we prove \cref{thm:bipartite}.  
The key missing element is \cref{lem:bipartiteHerm}, which tells us that in the bipartite case we can force the local matrices to be Hermititian without increasing the number of terms.

\begin{lemma}[Hermiticity in the bipartite case]\label{lem:bipartiteHerm} 
Let $P_1,\ldots, P_r\in \mc{M}_s$ and   
$Q_1,\ldots, Q_r \in \mc{M}_t$ be two sets of linearly independent matrices such that 
$$
\rho = \sum_{\alpha=1}^r P_\alpha \otimes Q_\alpha
$$ 
is Hermitian. Then $\rho$  can be expressed as 
$$
\rho = \sum_{\alpha=1}^r  A_\alpha \otimes B_\alpha 
$$
where $A_\alpha\in \mathrm{Her}_s$ and  $B_\alpha\in \mathrm{Her}_t$.  
\end{lemma}

\begin{proof} 
First note that a Hermitian decomposition $\rho=\sum_{\alpha=1}^{r'}  A_\alpha \otimes B_\alpha $ (with each $A_\alpha$ and $B_\alpha$ Hermitian, and possibly $r'\geq r$) can always be found, since $\Her_s \otimes \Her_t = \Her_{st}$. This holds simply because $\Her_s \otimes \Her_t \subseteq \Her_{st}$ and the two vector spaces have the same dimension, namely $s^2t^2$.

To see that $r'=r$, assume that  $\rho=\sum_{\alpha=1}^{r'}  A_\alpha \otimes B_\alpha $ is a minimal Hermitian decomposition, i.e.\ where $\{A_\alpha\}_\alpha$ are linearly independent over $\mathbb{R}$, and so are $\{B_\alpha\}_\alpha$. But this implies that they are also linearly independent over  $\mathbb{C}$, since consider $\lambda_\alpha\in \mathbb{C}$ and the equation 
$$
\sum_{\alpha} \lambda_\alpha A_\alpha = 0 .
$$
Summing and subtracting its complex conjugate transpose of this equation, 
%and using that $A_\alpha$'s are Hermitian and linearly independent, 
we see that the real and imaginary part of each $\lambda_\alpha$ needs to vanish. 
Thus  $r=r'$.
\end{proof}

%%------------------------------------------
Now we are ready to prove  \cref{thm:bipartite}. Let us first state the theorem again.

\begin{theorem-nn}[\cref{thm:bipartite}]
Let $\rho$ be a bipartite state 
$0\leqslant \rho\in \mc{M}_{d_1}\otimes \mc{M}_{d_2} $ 
where $d_1,d_2$ are arbitrary, such that 
$$
\rho = 
\sum_{\alpha=1}^2 
A_{\alpha} \otimes B_\alpha 
$$
where $A_\alpha\in \mc{M}_{d_1}$ and $B_\alpha\in \mc{M}_{d_2}$. 
Then $\rho$ is separable and can be written as 
$$
\rho = 
\sum_{\alpha=1}^2 
\sigma_{\alpha}\otimes \tau_\alpha
$$
where  $\sigma_\alpha \in \PSD_{d_1}$ and 
$\tau_\alpha \in \PSD_{d_2}$. 
%This means that $\seprank(\rho)=2$ and thus also $\purirank(\rho)=2$. 
\end{theorem-nn}

\begin{proof}
By \cref{lem:bipartiteHerm}, $\rho$ can be written with Hermitian local matrices, namely as 
$$
\rho =\sum_{\alpha=1}^{2}P_\alpha\otimes Q_\alpha 
$$
where $P_\alpha\in \Her_{d_1}$ and $Q_\alpha\in \Her_{d_2}$. 
By \cref{cor:main} \ref{cor:3}, there is a separable decomposition with only two terms,  
\begin{equation}\tag*{\qedhere}\rho=  \sum_{\alpha=1}^2 \sigma_\alpha \otimes \tau_\alpha. 
\end{equation}
\end{proof}

%%===================================
\subsection{Method to obtain the separable decomposition}
\label{ssec:method}

We now provide an explicit procedure to obtain the separable decomposition of a bipartite state with operator Schmidt rank two. 
We remark that another such procedure was provided in \cite[Theorem 3.44]{Ca17b}.
The procedure follows the proof of \cref{thm:bipartite}, as the latter is constructive. 
In the following we focus on presenting the method, rather than on justifying why it works. 
At the end of this section, we will apply this method to an example. 

First, the input and output to this method are as follows:
\begin{description}
\item[Input] A decomposition 
 of a bipartite positive semidefinite matrix 
$$
\rho= \sum_{\alpha=1}^2 P_\alpha\otimes Q_\alpha  \geqslant 0 . 
$$ 
where the $P_\alpha$'s and the $Q_\alpha$'s are linearly independent.\footnote{This could be, for example, the operator Schmidt decomposition of $\rho$.}
\item[Output] $\sigma_\alpha \geqslant 0$ and $\tau_\alpha \geqslant 0$ such that 
\be
\rho = \sum_{\alpha=1}^2\sigma_\alpha\otimes \tau_\alpha .
\label{eq:rhosepdecomp}
\ee 
\end{description}

The method works as follows:
\begin{description}
\item[Step 1] \label{im:step1} Making the matrices Hermitian. 

\begin{enumerate}[1.]
\im Express $P_\alpha = P_\alpha^0 + i P_\alpha^1$, where $P_\alpha^{0}$ and  $i P_\alpha^1$ are the Hermitian and antihermitian part of $P_\alpha$, respectively\footnote{Namely $P_\alpha^0 = (P_\alpha+P_\alpha^\dagger)/2$ and $P_\alpha^1 = -i(P_\alpha -P_\alpha^\dagger)/2$.}, and similarly for  $Q_\alpha$.

\im  Find two matrices of the set $\{P_1^0, P_1^1, P_2^0, P_2^1\}$ which are $\R$-linearly independent, i.e.\ that are not a multiple of each other. 

\im Express the other two matrices as a real linear combination of the two linearly independent matrices. Substitute this in $\rho$ and obtain an expression of the form 
\be
\rho= A\otimes C + B\otimes D \geqslant 0 
\label{eq:rhostep1}
\ee
where each matrix is Hermitian. 
\end{enumerate}

\item[Step 2] Finding the extreme rays $v,u\in\R^2$: 
\be
\textrm{S}(A,B) &=& \{(x,y)\in \R^2 | x A+y B\geqslant 0\}  \nn\\
&=& \textrm{convex-cone}\{v , u \}\subseteq\R^2.\nn
\ee

Denote the $i$th diagonal element of $C$, $D$ by $C_{ii}$, $D_{ii}$, respectively. Note that $C_{ii}A +D_{ii}B\geqslant 0$ by construction. There are two cases: 

\begin{enumerate}
\item[Case 1:] There is an $i$ such that $C_{ii}A + D_{ii}B> 0.$ Since this is positive definite, there is an invertible matrix $P$ such that $P^\dagger (C_{ii}A +D_{ii}B) P = I$. Set $$\tilde A=P^\dagger AP,\quad  \tilde B=P^\dagger BP$$ and let $\lambda_{\min}, \lambda_{\max}$ be the smallest and largest eigenvalue of $$C_{ii}\tilde B- D_{ii}\tilde A.$$  The extreme rays are 
$$
u = -(\lambda_{\min}C_{ii}+D_{ii},\lambda_{\min}D_{ii}-C_{ii}),$$
$$
v =  (\lambda_{\max}C_{ii}+D_{ii},\lambda_{\max}D_{ii}-C_{ii}).$$

\item[Case 2:] 
For all $i$, $C_{ii}A + D_{ii}B\ngtr 0$.  Then $(C_{ii}, D_{ii})$ is an extremal ray for any $i$. 
If there are two $i$'s whose corresponding vectors $(C_{ii}, D_{ii})$ are linearly independent, then we have  already found the two extremal rays.
If all $(C_{ii}, D_{ii})$  are linearly dependent, then we only have one extremal ray. In this case,
consider $$(C'_{ii}, D'_{ii}) = (C_{ii}, D_{ii}) \pm \epsilon (-D_{ii},C_{ii}).$$ For small $\epsilon$, the new point will be in the interior of $S(A,B)$. 
The kernel of $C'_{ii}A + D'_{ii}B$ is then precisely the joint kernel of $A$ and $B$. After splitting it off, the new matrices $A',B'$ will define the same spectrahedral cone and  $C'_{ii}A' + D'_{ii}B'>0$. Hence we are in Case 1.

%If we always find the same  the same extreme ray, then I go orthogonally to the left or right, and find an interior point. 
\end{enumerate}

Note that it is very unlikely that Case 2 of Step 2 happens, as the set of all $(C,D)\in\mathbb R^2$  satisfying  $C\cdot A + D\cdot B\geqslant0$ but $C\cdot A + D\cdot B\ngtr 0$ is a set of measure zero. 
In most cases there will be an $i$ which provides an interior point. 

\item[Step 3] Finding the positive semidefinite matrices $H_1, H_2$ such that  $A,B$
\be
C = u_1 H_1 + v_1 H_2, \nn\\
D = u_2 H_1 + v_2 H_2,\nn
\ee
where  $u=(u_1,u_2)$ and $v=(v_1,v_2)$. 
Since $u,v$ are linearly independent, this system of equations has a unique solution. 

Substituting in \eqref{eq:rhostep1} we obtain that the positive semidefinite matrices of \eqref{eq:rhosepdecomp} are
\be
&&\sigma_1 = u_1A+u_2B, \qquad \qquad\sigma_2 = v_1A+v_2B\nn\\
&&\tau_1= H_1, \quad 
\tau_2 = H_2. \nn
\ee

\end{description}

%%------------
We now illustrate this procedure with a simple example, namely a two-qubit state. In this case the negativity of the partial transpose is a sufficient criterion for separability. 
%, and thus we can easily check that it is separable. 
We use the standard conventions for qubits, 
where $I = |0\ra\la 0| + |1\ra\la 1|$, 
 $\sigma_z = |0\ra\la 0 | -|1\ra\la 1|$, and 
 $\sigma_x = |0\ra\la 1| + |0\ra\la 1|$.

\begin{example}[Two qubits correlated in $\sigma_x$ basis]\label{ex:xx}
Consider the two-qubit state 
$$
\rho = \frac{1}{2}(I\otimes I + \sigma_x \otimes \sigma_x).
$$
This has $\osr(\rho)=2$ and is thus separable. In this case the separable decomposition can be found by inspection, namely
\be
\rho = \frac{1}{2}(|+,+\ra \la +,+| +|-,-\ra \la -,-|  ),
\label{eq:rhosep+}
\ee 
where $|\pm \ra = (|0\ra\pm|1\ra)/\sqrt{2}$,  which indeed has $\seprank(\rho)=2$. 

We now apply the method above to this state to obain the separable decomposition. 
From now on we ignore normalization, i.e.\ we consider $2\rho$.  
\begin{description}
\im[Step 1] The matrices are already Hermitian. 
\im[Step 2] The spectrahedron defined by $I,\sigma_x$ is
$$
S(I, \sigma_x) = \{(a,b) | aI + b\sigma_x \geqslant 0\}
$$
The extreme rays are $u = (1,-1)$ and $v=(1,1)$. 
\im[Step 3] The positive semidefinite matrices $H_1,H_2$ that satisfy 
\be
I =  H_1 + H_2, \quad  
\sigma_x =  - H_1 +H_2 \nn
\ee
are $H_1 = |-\ra\la -|$ and  $H_2 = |+\ra\la +|$. Thus 
\be
&&\sigma_1 = |-\ra\la -| , \quad
\sigma_2= |+\ra\la +|, \nn \\
&&\tau_1 = I-\sigma_x , \quad
\tau_2 = I+\sigma_x \nn
\ee
which is indeed the decomposition of \eqref{eq:rhosep+}. 
\end{description}
\demo\end{example}

%%==================
\subsection{Context and implications of  \cref{thm:bipartite}}
\label{ssec:context}

In this section we put \cref{thm:bipartite} in the context of Ref.\ \cite{De19}, 
which presents 
a relation between decompositions of quantum states and decompositions of nonnegative matrices. 
Namely, Ref.\ \cite{De19} and \cite{De19d} show that for bipartite positive semidefinite matrices which are diagonal in the computational basis, 
\be
\label{eq:rhoclass}
\rho =\sum_{i,j} M_{i,j} |i,j\ra\la i,j|, 
\ee 
certain decompositions of $\rho$ correspond to  decompositions of the nonnegative matrix\footnote{That is, entrywise nonnegative. A Hermitian matrix with nonnegative eigenvalues is called positive semidefinite.} 
\be
M = \sum_{i,j} M_{i,j} |i\ra \la j| .
\label{eq:M}
\ee 
In particular, the operator Schmidt decomposition, 
the separable decomposition, 
and the local purification form of $\rho$
correspond to 
the singular value decomposition,
the nonnegative factorisation \cite{Ya91}, 
and the positive semidefinite factorisation \cite{Fi11}
of $M$, respectively. 
(There are also correspondences between the translationally invariant (t.i.) versions, 
 by which the 
 t.i.\ separable decomposition and the  t.i. local purification form of $\rho$
  correspond to 
the completely positive factorisation  \cite{Be15b}
and the completely positive semidefinite transposed factorisation of $M$ \cite{La15c}, respectively).
Moreover, each factorisation has an associated rank, and these correspondences imply that 
the operator Schmidt rank of $\rho$ equals the rank of $M$, 
the separable rank of $\rho$ equals the nonnegative rank of $M$, 
and 
the purification rank of $\rho$ equals the positive semidefinite rank of $M$. 
In a formula, if $\rho$ is of the form \eqref{eq:rhoclass} and $M$ of the form \eqref{eq:M}, then 
\be
\osr(\rho)&=&\rank(M) \nn\\
\seprank(\rho) &=&\rank_+(M) \label{eq:rel}\\
\purirank(\rho) &=&\psdrank(M) \nn
\ee

Now, it is known that: 

\begin{proposition}[\cite{Fa14}] \label{pro:rankclass}
 Let $M$ be a nonnegative matrix. 
Then $\rank(M) =2 \implies \rank_+(M)=
 \psdrank(M) =2.$
\end{proposition}

\cref{thm:bipartite} can be seen as a generalization of \cref{pro:rankclass} to the case that $\rho$ is not diagonal in the computational basis. 
And, as we will see in \cref{ssec:implranks}, \cref{thm:mainmpo} can be seen as a generalization of \cref{pro:rankclass}  to the multipartite case. 
The precise implications of \cref{thm:bipartite} for the ranks of  bipartite positive semidefinite matrices are the following. 

\begin{corollary}[of \cref{thm:bipartite}]\label{cor:ranks1}
Let $\rho$ be a bipartite positive semidefinite matrix,  
$0\leqslant \rho\in \mc{M}_{d_1}\otimes \mc{M}_{d_2} $
where $d_1,d_2$ are arbitrary. 
Then the following are equivalent:
\begin{enumerate}
\item[(i)] $\osr(\rho)=2.$
\item[(ii)] $\hosr(\rho)=2.$
\item[(iii)] $\seprank(\rho)=2.$
\end{enumerate}
Moreover, they imply that:
\begin{enumerate}
\item[(iv)] $\purirank(\rho)=2$ and $\rho$ is separable. 
\end{enumerate}
\end{corollary}

\begin{proof}
(i) implies (ii): is shown in \cref{lem:bipartiteHerm}.

(ii) implies (i):  since  $\osr(\rho) \leq \hosr(\rho)$ we only have to see that $\osr(\rho)=1$ cannot hold. If $\osr(\rho)=1$ then $\rho = A\otimes B$ for some matrices $A,B$, but since $\rho\geqslant 0$, $A$ and $B$ must both be positive semidefinite or both negative semidefinite. In either case they are both Hermitian, which shows that $\hosr(\rho)=1$, and contradicts the assumption $\hosr(\rho)=2$. 

(i) implies (iii): is shown in \cref{thm:bipartite}. 

(iii) implies (i): From \cref{thm:basic}  (iii) we have that $\osr(\rho)\leq \seprank(\rho)$. 
But $\osr(\rho)=1$ cannot hold, since this would imply that $\seprank(\rho)=1$ by \cref{thm:basic} (i), which contradicts the assumption that $\seprank(\rho)=2$.

(iii) implies (iv): From \cref{thm:basic} (iii) we have that $\purirank(\rho)\leq \seprank(\rho)$, thus we only have to see that $\purirank(\rho)=1$ cannot be true. By \cref{thm:basic} (i) this implies that  $\seprank(\rho)=1$, which contradicts the assumption that $\seprank(\rho)=2$.
\end{proof}

This leaves the following situation for bipartite positive semidefinite matrices. 
Note that we only need to analyse $\osr$, since $\osr(\rho) = \hosr(\rho)$ by \cref{lem:bipartiteHerm}. 
\begin{itemize}
\item If $\osr(\rho)=1$, then $\rho$ is a product state, and  $\hosr(\rho)=\seprank(\rho) = \purirank(\rho)=1$ by \cref{thm:basic} (i). 

\item  If $\osr(\rho)=2$ then  $\rho$  is a separable state, and $\hosr(\rho)=\seprank(\rho) = \purirank(\rho)=2$ by \cref{cor:ranks1}. 

\item Ref.\ \cite{Go12} provides a family of nonnegative matrices $\{ M_t \in \mathbb{R}^{t\times t}_{\geq 0} \}_{t\geq 3}$, for which $\rank(M_t) = 3$ for all $t$ but $\psdrank(M_t)$ diverges with $t$. 
Defining the bipartite state 
$$
\rho_t = \sum_{i,j=1}^t (M_t)_{i,j}|i,j\rangle\langle i,j| ,
$$ 
 and using relations \eqref{eq:rel}, this immediately implies that there is a family  $\{\rho_t\}$  for which $\osr(\rho_t)=3$ for all $t$ but $\purirank(\rho_t)$ diverges. By \cref{thm:basic} (iii), this implies that $\seprank(\rho_t)$ also diverges with $t$.

\end{itemize}

This is summarised in \cref{tab:overviewbipartite}, and this situation is to be compared with the multipartite case, summarised in \cref{tab:overview}. 

\begin{widetext}
\begin{center}
\begin{table}[thb]
\begin{tabular}{c|l}
$\osr(\rho)=\hosr(\rho)$ &  Bipartite state $\rho$\\ \hline 
1 & Product state  \\ 
2 &  Separable state, and separable rank and purification rank are both 2\\ 
3 &  If separable, purification rank 
can be unbounded, \\
& and thus separable rank can be unbounded too
\end{tabular}
\caption{Overview of the properties of bipartite positive semidefinite matrices $\rho$ with small operator Schmidt rank ($\osr$). Recall that the operator Schmidt rank equals the hermitian operator Schmidt rank ($\hosr$) by \cref{lem:bipartiteHerm}. }
\label{tab:overviewbipartite}
\end{table}
\end{center}
\end{widetext}

%%-------------

%%==================
\subsection{Further implications of \cref{thm:bipartite}}
\label{ssec:otherimpl}

We now derive some further implications of   \cref{thm:bipartite}, 
first concerning quantum channels, and then the entanglement of purification. 

First we focus on quantum channels, which are completely positive and trace preserving maps. 
Recall that 
every completely positive map $\mc{E}:\mc{M}_{d_2}\to \mc{M}_{d_1}$ is dual to a 
positive semidefinite matrix $0\leqslant\rho \in \mc{M}_{d_1}\otimes \mc{M}_{d_2}$ (sometimes called the Choi matrix of $\mc{E}$)  via the following relation \cite{Ja72}
\be
\rho = \frac{1}{\sqrt{d_1d_2}}(\mc{E}\otimes \textrm{id}) (\sum_{j,k=1}^{d_2}|j,j\ra\la k,k|)
\label{eq:cj}
\ee
where $\textrm{id}$ is the identity map. Moreover, $\mc{E}$ is trace preserving if and only if the Choi matrix satisfies that $\mathrm{tr_1}(\rho) = I/d_2$,  where $\mathrm{tr_1}$ indicates the trace over the first subsystem, and $I$ is the identity matrix of size $d_2$. 
A quantum channel $\mc{E}:\mc{M}_{d_2}\to \mc{M}_{d_1}$ is called \emph{entanglement breaking} if  $(\mc{E}\otimes I)(\sigma)$ is separable for all $0\leqslant \sigma \in \mc{M}_{d_1}\otimes \mc{M}_{d_2}$ \cite{Ho03}.
It is shown in Ref.\  \cite{Ho03} that $\mc{E}$ is entanglement breaking if and only if its Choi matrix $\rho$ is separable. 
Note that determining whether a quantum channel is entanglement breaking is NP-hard \cite{Gh10}.

\begin{corollary}[of \cref{thm:bipartite}] \label{cor:channels}
Let $d_1,d_2\geq 2$ be arbitrary and $\mc{E}: \mc{M}_{d_2} \to \mc{M}_{d_1}$ be a quantum channel such that the rank of $\mc{E}$ (i.e.\ the dimension of the image) is 2. Then $\mc{E}$ is entanglement breaking. 
\end{corollary}

\begin{proof}
If the dimension of the image of $\mc{E}$ is 2, then  $\mc{E}$ can be written as 
$\mc{E}(X) = \sum_{\alpha=1}^2 P_\alpha \tr(Q_\alpha^t X)$, where $t$ denotes transposition, and $P_1,P_2$ are linearly independent, as are $Q_1,Q_2$.
From \eqref{eq:cj} it is immediate to see that the Choi matrix is 
$$
\rho = \frac{1}{\sqrt{d_1d_2}}\sum_{\alpha=1}^2 P_\alpha\otimes Q_\alpha, 
$$
i.e.\ it has operator Schmidt rank 2.
By \cref{cor:main} \ref{cor:2}, $\rho$ is separable, 
 and thus $\mc{E}$ is entanglement breaking.  
\end{proof}

Note that the statement of \cref{cor:channels} also holds if $\mc{E}$ is not trace preserving, since in that case the Choi matrix also has operator Schmidt rank 2.

Now we turn to the implications of \cref{thm:bipartite} for the correlations in the system. 
We consider again a bipartite state, which we now denote by $\rho_{AB}$. The entanglement of purification $E_p(\rho_{AB})$ \cite{Te02} is defined as 
$$
E_p(\rho_{AB}) = \min_{\psi}  \{E(|\psi\ra_{A,A',B,B'} ))| \tr_{A'B'} |\psi\ra\la \psi| = \rho\} 
$$
where the entropy of entanglement $E(|\psi\ra_{A,A',B,B'} )$ is defined as 
$$
E(|\psi\ra_{A,A',B,B'} ) = S_1(\rho_{AA'}),  
$$ 
where $\rho_{AA'} = \tr_{BB'} |\psi\ra_{AA'BB'}\la \psi| $ and 
$S_1$ is the von Neumann entropy, $S_1(\rho) =-\tr(\rho\log_d\rho)$, where $\rho$ has size $d\times d$.  
%The entanglement of purification quantity is nonincreasing under local operations, but not under classical communication. 
The entanglement of purification is a measure of the amount of classical and quantum correlations of the system \cite{Te02}. 
It follows from \cref{thm:bipartite} that if a state has operator Schmidt rank two, then its entanglement of purification is very small: 

\begin{corollary}[of \cref{thm:bipartite}] \label{cor:corr}
Let $0\leqslant \rho \in \mc{M}_{d_1}\otimes \mc{M}_{d_2}$ be such that $\osr(\rho)=2$. 
Then its entanglement of purification satisfies $E_p(\rho)\leq \log_{d_1}2$. 
\end{corollary}

\begin{proof}
From \cref{cor:ranks1} we have that  $\osr(\rho)=2$ implies $\purirank(\rho)=2$. 
From \cite{De19} we know that $E_p(\rho)\leq \log_{d_1}(\purirank(\rho))$,  
 from which the result follows.
\end{proof}

%%===========================================
%%===========================================
\section{Multipartite states with Hermitian operator Schmidt rank two }
\label{sec:mpdo}

In this section we focus on multipartite states with Hermitian operator Schmidt rank two. 
First we prove \cref{thm:mainmpo} (\cref{ssec:thmmainmpo}), 
then analyse how to impose Hermiticity in the multipartite case (\cref{ssec:withc1}),
and finally   discuss the context and derive several implications of this result (\cref{ssec:implranks}).

\subsection{Proof of \cref{thm:mainmpo}}
\label{ssec:thmmainmpo}

In this section we prove  \cref{thm:mainmpo}. 
The main difference between the bipartite and the multipartite case (i.e.\ \cref{thm:bipartite} and \cref{thm:mainmpo}) 
is that in the bipartite case we need not assume that the matrices are Hermitian. 
This is because  \cref{lem:bipartiteHerm} allows to enforce the Hermiticity of the matrices without increasing the number of terms, i.e.\ it shows that  $\osr(\rho)=\hosr(\rho)$. 
Yet, this need not be true in the multipartite case, where we can only prove the looser relations between $\osr$ and $\hosr$ presented in \cref{ssec:withc1}. In particular, they do not guarantee that if $\osr(\rho)=2$ then $\hosr(\rho)=2$. 
Note that, ultimately, we need Hermitian matrices in order to use  \cref{cor:main} \ref{cor:3}. 

Let us now state and prove  \cref{thm:mainmpo}. 

\begin{theorem-nn}[\cref{thm:mainmpo}]
Let $\rho$ be  a positive semidefinite matrix, 
$
0\leqslant \rho\in \mc{M}_{d_1}\otimes \mc{M}_{d_2} \otimes \cdots \otimes\mc{M}_{d_n}
$ 
where $d_1,d_2,\ldots,d_n$ are arbitrary, such that
\be
\rho = \sum_{\alpha_1,\ldots, \alpha_{n-1}=1}^2 
A^{[1]}_{\alpha_1} \otimes A^{[2]}_{\alpha_1,\alpha_2}\otimes 
 \cdots 
 %\otimes A^{[n-1]}_{\alpha_{n-2},\alpha_{n-1}} 
 \otimes A^{[n]}_{\alpha_{n-1}}
\nn
\ee 
where $A^{[1]}_\alpha \in \Her_{d_1}$, $A^{[2]}_{\alpha,\beta} \in \Her_{d_2}$, \ldots, $A^{[n]}_{\alpha} \in \Her_{d_n}$.
Then $\rho$ is separable and can be written as 
$$
\rho = 
\sum_{\alpha_1,\ldots, \alpha_{n-1}=1}^2 
\sigma^{[1]}_{\alpha_1} \otimes \sigma^{[2]}_{\alpha_1,\alpha_2}\otimes 
%\sigma^{[3]}_{\alpha_2,\alpha_3} \otimes 
\cdots \otimes \sigma^{[n]}_{\alpha_{n-1}}
$$
where  $\sigma^{[1]}_\alpha \in \PSD_{d_1}$, 
$\sigma^{[2]}_{\alpha,\beta} \in \PSD_{d_2}$, 
$\ldots$, $\sigma^{[n]}_{\alpha} \in \PSD_{d_n}$.
\end{theorem-nn}

Note that we are assuming that $\sigma^{[1]}_1$ and $\sigma^{[1]}_2$ are linearly independent, and similarly for the other sites, since otherwise the sum over the corresponding index would only need to take one value.

\begin{proof}
 We prove the statement by induction on $n$. 
In the bipartite case we have 
$$
\rho = \sum_{\alpha=1}^2 A^{[1]}_{\alpha}\otimes A_{\alpha}^{[2]},
$$
where $A^{[1]}_{\alpha}$ and $A^{[1]}_{\alpha}$ are Hermitian. 
This is a particular case of  \cref{thm:bipartite}, and thus it follows that
$$
\rho = \sum_{\alpha=1}^2 \sigma^{[1]}_{\alpha}\otimes \sigma_{\alpha}^{[2]},
$$
where $\sigma^{[1]}_{\alpha}$ and $\sigma^{[2]}_{\alpha}$ are positive semidefinite. 
 For the induction step we also need the following property: 
 if another state admits a representation with the same $A^{[1]}_{\alpha}$ but  different $A^{[2]}_{\alpha}$, 
then its separable decomposition can be chosen with the same $\sigma^{[1]}_{\alpha}$, i.e.\ only $\sigma^{[2]}_{\alpha}$ is modified. 
This follows directly from the proof of \cref{thm:bipartite}. 
 
Now, for the induction step, assume that the statement holds for $n$ sites. 
We first write the positive semidefinite matrix $\rho$ on  $n+1$ sites as  
\be
\rho = \sum_{\alpha_{n}=1}^2 B_{\alpha_{n}}^{[1\ldots n]} \otimes A^{[n+1]}_{\alpha_n} 
\label{eq:rhon+1}
\ee
where
$$
B_{\alpha_{n}}^{[1\ldots n]} = \sum_{\alpha_1,\ldots, \alpha_{n-1}=1}^2 
A^{[1]}_{\alpha_1} \otimes A^{[2]}_{\alpha_1,\alpha_2} \otimes \cdots \otimes A^{[n]}_{\alpha_{n-1},\alpha_n} .
$$
Note that $A^{[n+1]}_{1}$, $A^{[n+1]}_{2}$ are Hermitian, and so are  $ B_{1}^{[1\ldots n]}$ and $ B_2^{[1\ldots n]}$. Thus, 
applying \cref{cor:main} \ref{cor:3} to $\rho$ in Eq.\ \eqref{eq:rhon+1}, we find that it is separable and can be written as 
\be 
\label{eq:rhon+1sep}
\rho =\sum_{\alpha_n=1}^2 \sigma^{[1\ldots n]}_{\alpha_n}\otimes \sigma_{\alpha_n}^{[n+1]}
\ee
where $\sigma^{[1\ldots n]}_{\alpha_n} \in \PSD_{d_1\cdot d_2\cdots d_n}$ and 
$\sigma_{\alpha_n}^{[n+1]}\in \PSD_{d_{n+1}}$. 
Moreover, $\sigma^{[1\ldots n]}_{\alpha_n}$ is a real linear combination of $B_{1}^{[1\ldots n]}$ and $B_{2}^{[1\ldots n]}$, 
\be
&&\sigma^{[1\ldots n]}_{\alpha_n} =  
\sum_{j=1}^2 \lambda_{\alpha_n j} B_{j}^{[1\ldots n]} \nn\\
&&=
 \sum_{\alpha_1,\ldots, \alpha_{n-1}}^2 
 A^{[1]}_{\alpha_1}  \otimes \cdots \otimes  
\left( \sum_{j =1}^2 \lambda_{\alpha_n j} A^{[n]}_{\alpha_{n-1},j}\right)
 \label{eq:sigma1n}
\ee
Thus, both $\sigma_{1}^{[1\ldots, n]}$ and $\sigma_{2}^{[1\ldots, n]}$ are positive semidefinite, defined on $n$ sites, and of Hermitian operator Schmidt rank $2$, as Eq.\ \eqref{eq:sigma1n} shows. 
Note also that  $\sigma_{1}^{[1\ldots, n]}$ and  $\sigma_{2}^{[1\ldots, n]}$ only differ 
at the $n$-th site. 
By the induction hypothesis, both $\sigma_{\alpha_n}^{[1\ldots, n]}$ have separable  rank 2 representations, using the same positive semidefinite matrices at the first $n-1$ sites.  Explicitly, 
$$
\sigma_{\alpha_n}^{[1\ldots n]} = \sum_{\alpha_1,\ldots, \alpha_{n-1}=1}^2 
\sigma^{[1]}_{\alpha_1}\otimes 
\sigma^{[2]}_{\alpha_1,\alpha_2}\otimes \cdots \otimes \sigma^{[n]}_{\alpha_{n-1},\alpha_n},
$$
where  $\sigma^{[1]}_\alpha \in \PSD_{d_1}$, 
$\sigma^{[2]}_{\alpha,\beta} \in \PSD_{d_2} \ldots$
Inserting this expression in $\rho$ of Eq.\ \eqref{eq:rhon+1sep} we obtain a separable decomposition of $\rho$ with separable rank 2. 

Finally, note that if only the  $A^{[n+1]}_{\alpha_n}$ in the initial representation of $\rho$ are changed (Eq.\ \eqref{eq:rhon+1}), then the  separable decomposition  that we have just constructed will only differ at the last site.
\end{proof}

%%=========

\subsection{Imposing Hermiticity in the multipartite case}
\label{ssec:withc1}

Here we analyse how to impose Hermiticity of the local matrices for multipartite states, and then prove a modified version of \cref{thm:mainmpo} in which the local matrices need not be Hermitian but linearly independent. 

As we have seen in \cref{lem:bipartiteHerm}, 
in the bipartite case we can impose Hermiticity of the local matrices without increasing the number of terms, that is, the operator Schmidt rank and its Hermitian counterpart coincide, $\hosr(\rho)=\osr(\rho)$. 
In the multipartite case we can only give the following bound, which we do not know  if it is tight.

\begin{proposition}[Hermiticity in the multipartite case]\label{pro:hermmulti}
Let $\rho$ be a Hermitian matrix in $ \mc{M}_{d_1} \otimes \cdots \otimes \mc{M}_{d_n}$, where $d_1,\ldots, d_n$ are arbitrary. 
Then 
$$ \osr(\rho)\leq \hosr(\rho)\leq 2^{n-1}\osr(\rho),$$  
except if $n=2$,  in which case $\hosr(\rho)=\osr(\rho)$.
\end{proposition}

\begin{proof}
The statement for $n=2$ corresponds to \cref{lem:bipartiteHerm}. 
For general $n$, 
we consider the MPDO form of $\rho$ [Eq.\ \eqref{eq:mpdo}], and express each matrix as a sum of its Hermitian and anti-Hermitian part, namely as 
$$
A_{\alpha}^{[l]} = \sum_{k=0}^1 i^k A_{\alpha}^{[l] k} 
$$
for $l=1,n$, and 
$$
A_{\alpha,\beta}^{[l]} = \sum_{k=0}^1 i^k A_{\alpha,\beta}^{[l] k} 
$$
for $1<l<n$. Substituting in the MPDO form of $\rho$ we obtain
$$
\rho=\sum_{k_1,\ldots, k_n=0}^1 i^{\mathbf{k}}
\sum_{\alpha_1,\ldots, \alpha_{n-1}=1}^D 
A^{[1] k_1}_{\alpha_1} \otimes 
A^{[2] k_2}_{\alpha_1,\alpha_2} \otimes 
\cdots \otimes
A^{[n] k_n}_{\alpha_{n-1}} ,
$$
where $\mathbf{k}= k_1+\cdots + k_n$. Since $\rho$ is Hermitian, all terms where $\mathbf{k}$ is odd must vanish. 
Thus the above expression contains $2^{n-1}$ terms, all of which are are Hermitian. Thus $\hosr(\rho)\leq 2^{n-1} \osr(\rho)$.
\end{proof}

With the help of an additional condition we can prove an analogous result to the bipartite case.\footnote{We thank an anonymous referee for suggesting this argument.} The following result is presented for finitely many matrices $\rho_i$, because this is needed in the induction step of the proof. Afterwards, we will prove a corollary for a single $\rho$.

%\begin{condition}\label{cond:c1}
%The Hermitian matrices $\rho_1,\ldots, \rho_k \in \mathcal{M}_{d_1}\otimes \mathcal{M}_{d_2} $ have rank $k$ representations, using the same linearly independent matrices at the first site. 
%\end{condition}

\begin{proposition}
\label{prop:multi} 
Let $$
 \rho_1,\ldots, \rho_r\in \mc{M}_{d_1}\otimes \mc{M}_{d_2} \otimes \cdots \otimes\mc{M}_{d_n}
$$ be Hermitian matrices
such that
\be
\rho_i = \sum_{\alpha_1,\ldots, \alpha_{n-1}=1}^k 
P^{[1]}_{\alpha_1} \otimes P^{[2]}_{\alpha_1,\alpha_2}\otimes 
 \cdots 
\otimes P_{\alpha_{n-2},\alpha_{n-1}}^{[n-1]} \otimes P^{[n]}_{\alpha_{n-1},i}
\nn
\ee 
where at each site the local matrices of $\rho_i$ are linearly independent. 
Then $\rho_1,\ldots, \rho_r$  can be written as 
$$
\rho_i = 
\sum_{\alpha_1,\ldots, \alpha_{n-1}=1}^k 
A^{[1]}_{\alpha_1} \otimes A^{[2]}_{\alpha_1,\alpha_2}\otimes 
%\sigma^{[3]}_{\alpha_2,\alpha_3} \otimes 
\cdots\otimes A_{\alpha_{n-2},\alpha_{n-1}}^{[n-1]}\otimes A^{[n]}_{\alpha_{n-1},i}
$$
where  all local matrices are Hermitian.
\end{proposition}

\begin{proof}The proof is by induction on $n$.
The case $n=2$ is precisely \cref{lem:bipartiteHerm}, where the fact that the Hermitian matrices at site one can be chosen the same for all $\rho_i$ follows from choosing  a Hermitian basis of ${\rm span}\{P_1^{[1]},\ldots, P_k^{[1]}\}$.

For a general $n$, we proceed as in the proof of \cref{thm:mainmpo}. It is easy to see that linear independence of the local matrices at each site is maintained throughout the construction. 
\end{proof}

\begin{corollary}
Let $
0\leqslant \rho\in \mc{M}_{d_1}\otimes \mc{M}_{d_2} \otimes \cdots \otimes\mc{M}_{d_n}
$ be positive semidefinite with 
\be
\rho = \sum_{\alpha_1,\ldots, \alpha_{n-1}=1}^2 
P^{[1]}_{\alpha_1} \otimes P^{[2]}_{\alpha_1,\alpha_2}\otimes 
 \cdots 
\otimes P_{\alpha_{n-2},\alpha_{n-1}}^{[n-1]} \otimes P^{[n]}_{\alpha_{n-1}}
\nn
\ee 
where at each site the local matrices are linearly independent. Then $\rho$ is separable with $\seprank(\rho)=2.$
\end{corollary}
\begin{proof}Clear from \cref{prop:multi} and \cref{thm:mainmpo}.
\end{proof}

%%=========
\subsection{Context and implications of \cref{thm:mainmpo}}
\label{ssec:implranks}

As mentioned in \cref{ssec:context}, \cref{thm:mainmpo} can be seen as a generalization of \cref{pro:rankclass} to the case that $\rho$ is a general quantum state (i.e.\ not diagonal in the computational basis), and it is multipartite. In particular, \cref{thm:mainmpo} implies the following relations among ranks for the multipartite case. 

\begin{corollary}[of \cref{thm:mainmpo}] \label{cor:ranks} 
Let $\rho$ be a multipartite positive semidefinite matrix, 
$$
0\leqslant \rho\in \mc{M}_{d_1}\otimes \mc{M}_{d_2} \otimes \cdots \otimes\mc{M}_{d_n} , 
$$ 
where $d_1,d_2,\ldots,d_n$ are arbitrary. 
Then the following are equivalent: 
\begin{itemize}
\item[(i)]$\hosr(\rho) = 2.$
\item[(ii)]  $\seprank(\rho)=2$. 
\end{itemize}
Moreover, they imply the following: 
\begin{itemize}
\item[(iii)] $ \purirank(\rho) = 2$ and $\rho$ is separable.
\end{itemize}
\end{corollary}

\begin{proof}
(i) implies (ii):  is shown in \cref{thm:mainmpo}.

(ii) implies (i): From  \cref{thm:basic} we have that  $\hosr(\rho)\leq\seprank(\rho)$, and $\hosr(\rho)=1$ implies $ \seprank(\rho) = 1$, which contradicts the assumption. 

(ii) implies (iii): From \cref{thm:basic} (iii) we have that $\purirank(\rho)\leq\seprank(\rho)$, and from  \cref{thm:basic} (i)
$\purirank(\rho)=1$ implies that $\seprank(\rho)=1$, which contradicts the assumption. 
\end{proof}

%%-----------------
This leaves the following situation for multipartite positive semidefinite matrices of small operator Schmidt rank (summarised in \cref{tab:overview}):
\begin{itemize}
\item If $\osr(\rho)=1$, then $\rho$ is a product state, and  $\hosr(\rho)=\seprank(\rho) = \purirank(\rho)=1$ by \cref{thm:basic} (i). 

\item If $\osr(\rho)=2$ and $\hosr(\rho)=2$, then  $\rho$  is a separable state, and $\osr(\rho)=\seprank(\rho) = \purirank(\rho)=2$ by \cref{cor:ranks}. 

\item If $\osr(\rho)=2$ and $\hosr(\rho)>2$, we cannot conclude anything about the state. (Note that from \cref{pro:hermmulti} we only know that, if $\osr(\rho)=2$, then $2\leq \hosr(\rho)\leq 2^n$, where $n$ is the number of sites, i.e.\ $\rho\in \mc{M}_{d_1}\otimes \cdots \otimes \mc{M}_{d_n}$).

\item Ref.\ \cite{De13c} provides a family of states 
$$
\{\rho_n \in \mc{M}_{2} \otimes \cdots \otimes \mc{M}_{2} \: (n \textrm{ times}) \}_{n\geq 1}
$$ 
which are diagonal in the computational basis (thus, separable) for which $\osr(\rho_n)=3$ for all $n$, and $\purirank(\rho_n)$ diverges with $n$. By \cref{thm:basic} (iii), this implies that $\seprank(\rho_n)$ also diverges with $n$.
\end{itemize}

In words, the case $\hosr(\rho)=2$ behaves very similarly to the trivial case $\hosr(\rho)=1$, 
as one can determine the  kind of state $\rho$ is, and one  can upper bound the amount correlations. 
On the other hand, the separation between $\osr(\rho)$ and $\purirank(\rho)$ appears for a very small value of $\osr$, namely $\osr(\rho)=3$. 

\begin{widetext}
\begin{center}
\begin{table}[htb]
\begin{tabular}{c|c|l}
$\osr(\rho)$ & $\hosr(\rho)$&  Multipartite state $\rho$ \\ \hline 
1 &1 &Product state  \\ 
2 & 2& Separable state, and separable rank and purification rank are both 2\\
2 &  $> 2$&?\\
3 & $\geq 3$  &If separable, purification rank 
can be unbounded, \\
& &and thus separable rank can be unbounded too
\end{tabular}
\caption{Overview of the properties of multipartite positive semidefinite matrices $\rho$ with small operator Schmidt rank ($\osr$) and hermitian operator Schmidt rank ($\hosr$).}
\label{tab:overview}
\end{table}
\end{center}
\end{widetext}

%%===================================
\section{Conclusions \& Outlook}
\label{sec:concl}

We have presented two main results.
First, we have shown that any bipartite positive semidefinite matrix of operator Schmidt rank two is separable, and admits a separable decomposition with  two  positive semidefinite matrices per site (\cref{thm:bipartite}). 
We have also provided a step-by-step method to obtain this separable decomposition (\cref{ssec:method}), 
and have drawn some consequences of this result concerning the relation among the ranks (\cref{cor:ranks1}, see also the summary of \cref{tab:overviewbipartite}), 
for quantum channels  (\cref{cor:channels}), and for the amount of correlations of the state (\cref{cor:corr}). 
While this result was already known \cite{Ca14b}, we have provided a new proof thereof. In particular, we have leveraged a relation between minimal operator systems and separable states (\cref{thm:main}) and its implications (\cref{cor:main}).

Second, we have shown that any multipartite positive semidefinite matrix of Hermitian operator Schmidt rank two is separable, and it has separable rank two (\cref{thm:mainmpo}).  
We have drawn consequences of this result for the ranks (\cref{cor:ranks}, see also the summary of \cref{tab:overview}).

This work leaves a number of interesting open questions. 
A first question concerns the generalisation of these results to the case of operator Schmidt rank three.  
Bipartite states $\rho\in \mathcal{M}_{d}\otimes \mathcal{M}_{d'}$  with operator Schmidt rank three must be separable  if $d=2$ or $d'=2$ \cite{Ca14b}, but this is no longer true in the case that $d>2$ and $d'>2$. 
Further results for these states are shown in \cite[Corollary III.5]{Ch18}.  
It would be interesting to analyse the  consequences of these results for the multipartite case, and compare them to   \cref{tab:overview}.

A second  question is to clarify the relation between the operator Schmidt rank and the hermitian operator Schmidt rank in the multipartite case. In particular, it would be interesting to find out whether \cref{pro:hermmulti} is tight or can be improved, and  
 to find direct relations between the $\hosr$ and the separable rank and/or the purification rank. 

On a broader perspective, it would be worth studying which of these results extends to the translationally invariant (t.i.) case. 
Ref.\ \cite{De19} presents the t.i.\ versions of the MPDO form, the separable form, and the local purification form. 
For bipartite states which are diagonal in the computational basis, 
the ranks associated to each of these decompositions 
correspond to well-studied ranks of nonnegative matrices (such as the cp rank and the cpsd rank), 
as mentioned in \cref{ssec:context}. Exploiting these relations may allow us to generalize results about ranks in that case, too, and gain further insight into the decompositions of t.i. quantum states. 

More generally, this paper illustrates how fruitful the connection between quantum states and (free) spectrahedra can be. The latter provide a novel set of techniques that can shed further light into the decompositions of quantum states. This is being explored in \cite{Ne19}.

%%===================================
\emph{Acknowledgements}. 
We thank an anonymous referee for suggesting  \cref{prop:multi} and bringing Refs.\ \cite{Ca17b,Gi18b} to our attention.
We thank A. M\"uller--Hermes, M.\ Studi\'nski and N.\ Johnston for bringing  Refs.\ \cite{Ca14b,Ca15b} to our attention. T.\ D.\ and T.\ N.\ are supported by the Austrian Science Fund FWF through project P 29496-N35.

%%==============================================
%\bibliographystyle{alphaurl}
%\bibliography{/Users/gemmadelascuevas/Dropbox/Gemma/Special-files/all-my-bibliography.bib}

\newcommand{\etalchar}[1]{$^{#1}$}

\end{document}